\documentclass{ifacconf}
\usepackage{graphicx}
\usepackage{natbib}
\usepackage{amsmath}
\usepackage{amssymb}
\usepackage{algorithmicx}
\usepackage{algorithm}
\usepackage{float}
\usepackage{caption}
\usepackage{textcomp}
\usepackage{epstopdf}
\usepackage{enumerate}
\usepackage{xcolor} 
\usepackage{arydshln}
\newtheorem{assumption}{\bf Assumption}
\newtheorem{remark}{\bf Remark}
\newtheorem{lemma}{\bf Lemma}
\newtheorem{theorem}{\bf Theorem}
\newtheorem{definition}{\bf Definition}
\newcommand{\smallsum}[1][]{%
  \mathop{\scalebox{0.7}{$\displaystyle\sum\nolimits_{#1}$}}%
}
\allowdisplaybreaks
\begin{document}
\begin{frontmatter}

\title{Privacy-Preserving Distributed Stochastic Optimization with Homomorphic Encryption and Heterogeneous Stepsizes} 

\thanks[footnoteinfo]{This work was supported in part by the National Natural Science Foundation of China under Grant 62203359, and in part by the China Postdoctoral Science Foundation under Grant 2025M774362.}

\thanks[footnoteinfo]{This is the full version of the paper accepted to the 23rd IFAC World Congress,
Busan, Republic of Korea, August 23-28, 2026. This version includes all proofs omitted from the conference proceedings due to page limitations.}

\author{Haoqiang Zhou, Chi Chen, Yongfeng Zhi, Huan Gao*} 
\address{School of Automation, Northwestern Polytechnical University, Xi'an 710129, China (e-mail: huangao@nwpu.edu.cn).}

\begin{abstract}
Distributed stochastic optimization enables multi-agent collaboration in applications such as distributed learning and sensor networks, but also raises critical privacy concerns due to the involvement of sensitive data. While existing privacy-preserving approaches often face limitations in balancing accuracy with efficiency, we propose a novel distributed stochastic gradient descent algorithm that integrates Paillier homomorphic encryption with heterogeneous and time-varying random stepsizes. The proposed algorithm provides inherent privacy protection against both internal honest-but-curious agents and external eavesdroppers, without relying on any trusted neighbors. Furthermore, we incorporate an attenuation factor to effectively mitigate quantization error induced by the encryption process, ensuring almost sure convergence to the optimal solution while maintaining privacy preservation. Numerical simulations demonstrate the effectiveness and efficiency of the proposed approach.
\end{abstract}


\end{frontmatter}

\section{Introduction}

Distributed stochastic optimization enables a group of agents to collaboratively optimize a global objective function through local computation and communication, using stochastic (noisy) local gradients. Due to its broad applicability, this field has garnered significant attention in areas such as distributed machine learning \citep{Tsianos2012}, sensor networks \citep{rabbat2004distributed}, and large-scale data analytics \citep{daneshmand2015hybrid}.

The theoretical foundations can be traced to early work on stochastic approximation in centralized frameworks \citep{robbins1951stochastic} and decentralized stochastic gradient descent \citep{tsitsiklis1986distributed}. Subsequent research has evolved along two branches: convex optimization \citep{Nedic2009Multi-agent, nedic2016stochastic} and nonconvex optimization \citep{bianchi2013convergence, Xin2021}. While early gradient descent methods required diminishing stepsizes for exact convergence \citep{yuan2016convergence}, later gradient-tracking techniques enabled faster convergence with constant stepsizes \citep{Pu2018tracking}. Based on the standard assumption of an unbiased gradient oracle, further developments have integrated network topology, adaptive stepsizes, and variance reduction to improve convergence speed and robustness to data heterogeneity \citep{Li2025Homogeneity}. More recent work has been extended to settings with biased gradient oracles \citep{Jiang2025unbiased}.

In distributed stochastic optimization, participating agents explicitly exchange their state variables and/or local gradients to reach the global optimum. However, this introduces significant privacy risks as sensitive information (e.g., medical records and financial transactions) is often inherent in individual agents' local objective functions and their gradients \citep{Wang2023privacy}. Consequently, this explicit information exchange inevitably leads to severe privacy leakage, which is unacceptable in safety-critical scenarios. For example, in distributed machine learning, participating agents only share model parameters and/or local gradients while keeping raw data local, an approach traditionally considered sufficient for protecting data privacy. Nevertheless, recent research works have revealed that the information transmitted among participants is capable of being utilized to recover the initial data in high-quality forms. Typical examples include pixel-perfect restoration for visual data and accurate token-level recovery for textual content \citep{zhu2019deep}.

Given the growing emphasis on privacy, developing privacy-preserving solutions for distributed stochastic optimization has become increasingly urgent. While numerous methods exist, most focus on deterministic gradient settings \citep{gao2023dynamics, Huang2024differential}, with relatively few contributed to stochastic optimization. Early approaches mainly relied on differential privacy \citep{abadi2016deep}. However, persistently injecting independent noise inherently degrades accuracy, leading to a privacy-accuracy trade-off. To circumvent this, \cite{gade2018privacy} leveraged network structure to generate spatially correlated noise for gradient obfuscation. While maintaining exact convergence, the method requires each agent to have enough neighbors free from adversarial collusion, thus reducing privacy protection. Unlike noise-injection-based approaches, \cite{wang2023decentralized} employed heterogeneous stepsizes and manipulated interaction dynamics to obscure gradients, ensuring privacy without sacrificing accuracy. An aggressive quantization based approach was developed in \cite{Wang2023privacy}, which masks the transmitted data and thereby provides privacy protection in distributed stochastic optimization. However, due to the use of homogeneous stepsizes, this approach cannot protect the privacy of an individual agent with only one neighbor or whose neighbors collude with an external eavesdropper. Hardware-based solutions like trusted enclaves \citep{ohrimenko2016oblivious, tramer2018slalom} have also been explored, however, they are not directly suited for preventing multiple data providers from inferring mutual data during stochastic optimization.

Motivated by these limitations and inspired by the cryptographic mechanism in \cite{Gao2019Secure}, we propose a novel privacy-preserving distributed stochastic optimization algorithm for undirected networks based on Paillier homomorphic cryptography and heterogeneous stepsizes. By leveraging partially homomorphic cryptography to embed privacy in pairwise interaction dynamics, our algorithm prevents plaintext leakage from external eavesdroppers capable of intercepting inter-agent communication. Furthermore, the proposed algorithm employs heterogeneous and time-varying random stepsizes, introducing perturbations to obscure the true gradient information. In contrast to many existing approaches \citep{gade2018privacy, zhang2019admm, Wang2023privacy}, our algorithm enables privacy preservation without assuming the existence of trusted neighbors; i.e., privacy is guaranteed even when all neighbors of a particular agent are adversaries colluding with external eavesdroppers. To counteract quantization error introduced by encryption, we incorporate a time-varying attenuation factor that progressively suppresses quantization error across iterations, guaranteeing almost sure convergence to the optimal solution while preserving privacy.

\textbf{Notations:} Let $\mathbb{R}$ and $\mathbb{Z}^+$ denote the sets of real numbers and positive integers, respectively, and let $\mathbb{R}^d$ be the $d$-dimensional Euclidean space. $\mathbf{1}_n$ and $\mathbf{0}_n$ represent the all-ones and all-zeros vectors in $\mathbb{R}^n$. The inner product is denoted by $\langle \cdot, \cdot \rangle$, while $\|\cdot\|$ represents the Euclidean norm for vectors and the spectral norm for matrices. The Frobenius norm is denoted by $\|\cdot\|_F$. $\lfloor \cdot \rfloor$ represents the floor (round down) operation, and $\otimes$ represents the Kronecker product. For a vector $\mathbf{x}$ and a matrix $X$, their $i$-th and $(i, j)$-th entries are denoted by $x_i$ and $x_{ij}$, respectively. $\mathbb{E}[x \, | \, \mathcal{F}]$ is the conditional expectation of random variable $x$ given the $\sigma$-algebra $\mathcal{F}$, and $\mathbb{D}[x]$ denotes $x$'s variance.

\section{Preliminaries}
\subsection{Problem Formulation}

We consider a network of $n$ agents, indexed by the set $\mathcal{V}=\{1, \ldots, n\}$, interacting over an undirected graph $\mathcal{G}=(\mathcal{V}, \mathcal{E})$. $\mathcal{E} \subseteq \mathcal{V} \times \mathcal{V}$ is the edge set representing communication links. For each agent $i \in \mathcal{V}$, its neighbor set is denoted by $\mathcal{N}_i = \{j \in \mathcal{V} \, | \, (i, j) \in \mathcal{E}\}$.

\begin{assumption}\label{ass:graph}
	$\mathcal{G}$ is a connected graph: any two agents $i, j \in \mathcal{V}$ are path-linked.
\end{assumption}
The agents collaboratively solve the following problem:
\begin{equation}\label{Fd}
	\min _{\mathbf{x} \in \mathbb{R}^d} F(\mathbf{x})=\frac{1}{n} \sum_{i=1}^n f_i(\mathbf{x}), 
\end{equation}
where $\mathbf{x} \in \mathbb{R}^d$ is the global optimization variable, and $f_i(\mathbf{x})$ is the local objective function privately accessible to agent $i$. The local function is often defined as the expected risk:
\begin{equation}
	f_i(\mathbf{x}) \triangleq \mathbb{E}_{ \boldsymbol{\eta}_i \sim \mathcal{O}_i}[l_i(\mathbf{x}, \boldsymbol{\eta}_i)], 
\end{equation}
where $l_i: \mathbb{R}^d \times \mathbb{R}^m \rightarrow \mathbb{R}$ denotes the local loss, and $\mathcal{O}_i$ characterizes the distribution of $\boldsymbol{\eta}_i \in \mathbb{R}^m$ (e.g., a data sample).

Since the distribution $\mathcal{O}_i$ is typically unknown and the volume of data is large, agent $i$ approximates the expected risk by maintaining a local copy of the global variable $\mathbf{x}_i \in \mathbb{R}^d$ and computing a stochastic (noisy) gradient. We denote the stochastic gradient estimator of $\nabla f_i(\mathbf{x}_i)$ at iteration $k$ as $\mathbf{g}_i^k(\mathbf{x}_i, \boldsymbol{\eta}_i)$, or simply $\mathbf{g}_i^k$. 

\begin{assumption}\label{ass:Ff}
	The global function $F$ is convex, and each $f_i$ has $L$-Lipschitz continuous gradients: $\|\nabla f_i(\mathbf{x})-\nabla f_i(\mathbf{y})\| \leq L\|\mathbf{x}-\mathbf{y}\|, \forall \mathbf{x}, \mathbf{y} \in \mathbb{R}^d$. Moreover, the stochastic gradient $\mathbf{g}_i^k$ satisfies $\mathbb{E}[\mathbf{g}_i^k]=\nabla f_i(\mathbf{x}_i)$ and $\mathbb{E}[\|\mathbf{g}_i^k-\nabla f_i(\mathbf{x}_i)\|^2] \leq \sigma_g^2$.
\end{assumption}

\subsection{Paillier Homomorphic Encryption}

The Paillier cryptosystem \citep{Paillier1999} consists of key generation, encryption, and decryption. Each agent $i$ holds a public key $k_{pi}$ and a private key $k_{si}$. Let $p$ be a plaintext and $q$ its ciphertext, with encryption $q = E_i(p)$ and decryption $p = D_i(q)$. Paillier supports additive homomorphism: $E_i(p_1)E_i(p_2) = E_i(p_1 + p_2)$.

Since Paillier encryption operates on binary strings, data must be quantized before transmission. Consider agent $i$'s state vector $\mathbf{x}_i = [x_{i1}, \ldots, x_{id}]^T \in \mathbb{R}^d$. Quantization is performed element-wise as:
\begin{equation}\label{Qq}
	Q(\mathbf{x}_i)= [Q( x_{i1}), \ldots, Q(x_{id})]^T, 
\end{equation}
\begin{equation}\label{eq:Q}
	Q(x_{il})= \begin{cases}
		\left\lfloor \frac{x_{il}}{\delta} \right\rfloor +1 \quad \text{with probability} \, P = \frac{x_{il}} {\delta} - \left\lfloor \frac{x_{il}}{\delta} \right\rfloor, \\
		\left\lfloor \frac{x_{il}}{\delta} \right\rfloor \quad \text{with probability} \, P = 1 - \frac{x_{il}}{\delta} + \left\lfloor \frac{x_{il}}{\delta} \right\rfloor, 
	\end{cases}
\end{equation}
where $l = 1, \ldots, d$, $\delta$ is a sufficiently small positive constant representing the quantization precision, and $\lfloor \cdot \rfloor$ denotes the floor (round down) operation. The inversely quantized value is denoted by $\tilde{x}_{il} \triangleq \delta Q(x_{il})$, satisfying
\begin{equation}\label{eq_inverse_quan}
	\mathbb{E}[\tilde{x}_{il}]=x_{il}, \quad \mathbb{D}[\tilde{x}_{il}] \leq\frac{1}{4}\delta^2, 
\end{equation}
which shows that the quantization is unbiased and has bounded variance.

\section{The Proposed Algorithm and Performance Analyses}
\subsection{Privacy-Preserving Algorithm Design}

To motivate our design, we first discuss the privacy vulnerability of conventional distributed stochastic gradient descent algorithm. As presented in \cite{Nedic2009Multi-agent}, its update rule is:
\begin{equation}\label{eq:SGD}
	\mathbf{x}_i^{k+1} = \mathbf{x}_i^k + \sum_{j \in \mathcal{N}_i} w_{ij}( \mathbf{x}_j^k - \mathbf{x}_i^k ) - \alpha^k \mathbf{g}_i^k, 
\end{equation}
where $\mathbf{x}_i^k \in \mathbb{R}^d$ is the state of agent $i$ at iteration $k$, $\alpha^k$ is the stepsize, and $w_{ij}>0$ is the coupling weight between agent $i$ and its neighbor $j$. The publicly transmitted states and homogeneous stepsize $\alpha^k$ allow an attacker to easily infer the gradient $\mathbf{g}_i^k$, leading to significant privacy leakage.

To address this vulnerability, we adopt the encryption-based mechanism in \cite{Gao2019Secure} and employ Paillier homomorphic encryption to secure information exchange. The protection mechanism begins with decomposing the coupling weight $w_{ij}$ as $w_{ij} \triangleq w_{i \rightarrow j} w_{j \rightarrow i}$, where $w_{i \rightarrow j}$ and $w_{j \rightarrow i}$ are independent values that privately chosen by agent $i$ and agent $j$, respectively. Specifically, each decomposed coupling weight is randomly chosen from the set $\{ k \delta \, | \, k \in \mathbb{Z}^+, \, 1\leq k \leq \lfloor \frac{1}{\delta} \rfloor \}$, ensuring it is an integer multiple of the quantization precision $\delta$. Under this construction, the inverse quantization results of the decomposed coupling weights exactly recover the original values with probability $1$ (i.e., $\tilde{w}_{i \rightarrow j} = w_{i \rightarrow j}$). This makes the quantization process deterministic (with zero variance), a property that is crucial for ensuring the bounded variance of the exchanged states in the subsequent analysis.

Before transmission, the interaction term undergoes a two-step process: quantization via \eqref{Qq} and \eqref{eq:Q} followed by Paillier encryption. The detailed exchange protocol between neighboring agents $i$ and $j$  is depicted in Fig. \ref{Paillier encrypt}.

\begin{figure}[h]
	\centering
	\includegraphics[width=1\columnwidth]{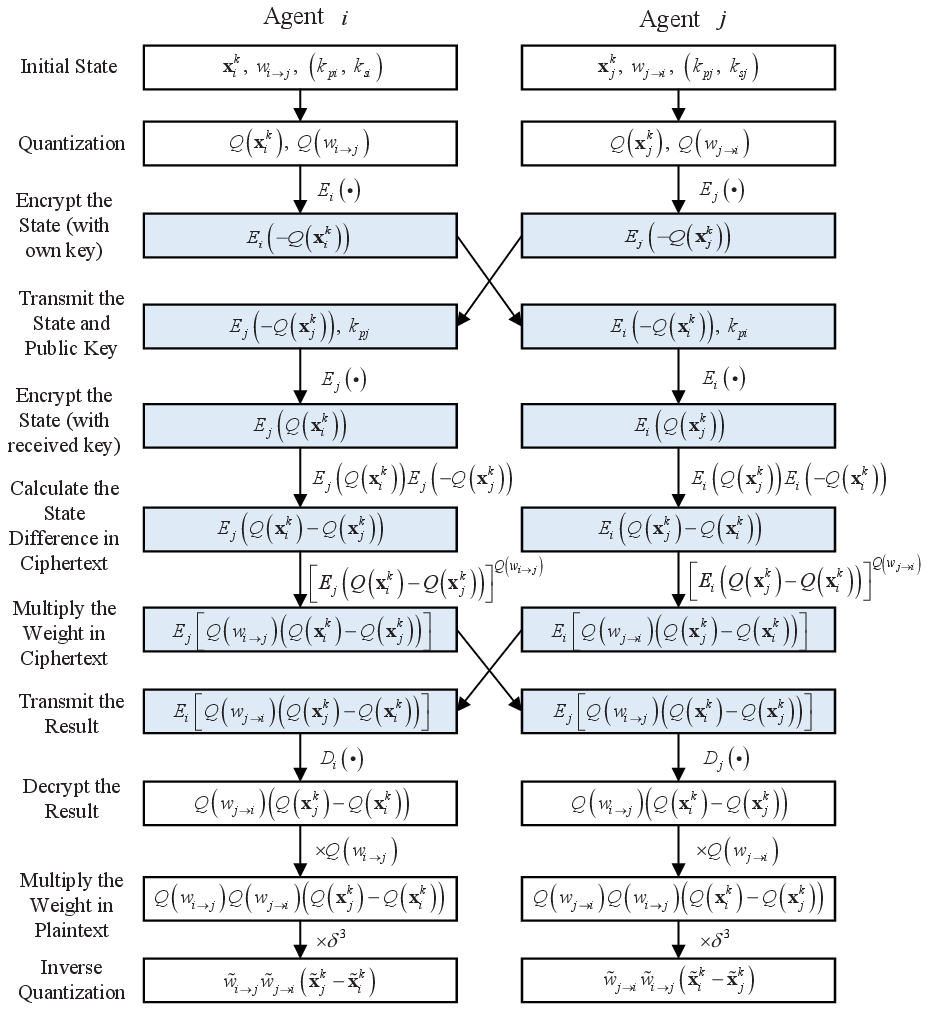}
	\caption{Paillier-based secure interaction protocol.}
	\label{Paillier encrypt}
\end{figure}

After receiving the encrypted data from neighbor $j$, agent $i$ decrypts and dequantizes it to obtain $\tilde{w}_{i \rightarrow j} \tilde{w}_{j \rightarrow i} ( \tilde{\mathbf{x}}_j^k-\tilde{\mathbf{x}}_i^k )$. Since quantization introduces errors between $w_{i \rightarrow j} w_{j \rightarrow i}$ $(\mathbf{x}_j^k - \mathbf{x}_i^k)$ and $\tilde{w}_{i \rightarrow j} \tilde{w}_{j \rightarrow i}(\tilde{\mathbf{x}}_j^k - \tilde{\mathbf{x}}_i^k)$, we introduce an attenuation factor $\gamma^k$ to suppress error accumulation and ensure convergence. Furthermore, each agent $i$ uses a private diagonal stepsize matrix $\Lambda _i^k \triangleq \operatorname{diag}\{ \lambda _{i1}^k, \ldots , \lambda _{id}^k \}$ to mask its true gradient $\mathbf{g}_i^k$, enhancing privacy.

\begin{assumption}\label{ass:step}
	For each agent $i\in\mathcal{V}$, its stepsize matrix $\Lambda_i^k$ satisfies $\mathbb{E}[\Lambda_i^k]=\lambda_i^k I_d$ and $\mathbb{D}[\Lambda_i^k]=(\sigma_i^k)^2 I_d$.
\end{assumption}

The proposed algorithm is summarized in Algorithm \ref{alg}.

\begin{algorithm}[h]
	\caption{The proposed privacy-preserving distributed stochastic gradient descent algorithm}
	\label{alg}
	\begin{algorithmic}
		\State \textbf{Parameter:} Attenuation factor $\gamma ^k > 0$.
		\State \textbf{Initialization:} Each agent $i$ randomly chooses coupling weights $w_{i \rightarrow j}$ for all $j \in \mathcal{N}_i$, and initializes $\mathbf{x}_i^0 \in \mathbb{R}^d$. Then it generates public key $k_{pi}$ and private key $k_{si}$, and sends $k_{pi}$ to all neighbors $j \in \mathcal{N}_i$.
		\State \textbf{For each iteration } $k=0, 1, \ldots$, \textbf{do}
		\State \quad (1) Agent $i$ generates the random stepsize $\Lambda _i^k \triangleq \operatorname{diag}\{ \lambda_{i1}^k, \ldots, \lambda_{id}^k \}$ following Assumption \ref{ass:step}.
		\State \quad (2) Agent $i$ executes the Paillier encrypted interaction protocol with its neighbor $j \in \mathcal{N}_i$ to obtain $\tilde{w}_{i \rightarrow j}\tilde{w}_{j \rightarrow i}(\tilde{\mathbf{x}}_j^k-\tilde{\mathbf{x}}_i^k)$.
		\State \quad (3) Agent $i$ updates its state via:
		\begin{equation}\label{x}
			\mathbf{x}_i^{k+1}=\mathbf{x}_i^k+{{\gamma }^k}\sum\limits_{j\in {{\mathcal{N}}_i}}{\tilde{w}_{i\to j}\tilde{w}_{j\to i}( \tilde{\mathbf{x}}_j^k-\tilde{\mathbf{x}}_i^k )}-\Lambda _i^k\mathbf{g}_i^k.
		\end{equation}
	\end{algorithmic}
\end{algorithm}

For each agent $i$, define the accumulated quantization error at iteration $k$ as $\boldsymbol{\xi}_i^k \triangleq \sum\nolimits_{j\in {{\mathcal{N}}_i}} \tilde{w}_{i \rightarrow j} \tilde{w}_{j \rightarrow i} ( \tilde{\mathbf{x}}_j^k - \tilde{\mathbf{x}}_i^k ) - \sum\nolimits_{j \in \mathcal{N}_i} w_{i \rightarrow j} w_{j \rightarrow i} ( \mathbf{x}_j^k - \mathbf{x}_i^k )$, then the update rule in \eqref{x} can be rewritten as
\begin{equation*}\label{xi}
	\mathbf{x}_i^{k+1} = \mathbf{x}_i^k + \gamma^k \sum_{j \in \mathcal{N}_i} w_{i \rightarrow j} w_{j \rightarrow i} ( \mathbf{x}_j^k - \mathbf{x}_i^k ) + \gamma^k \boldsymbol{\xi}_i^k - \Lambda_i^k \mathbf{g}_i^k.
\end{equation*}
From $\mathbb{E}[\tilde{w}_{i \rightarrow j} \tilde{w}_{j \rightarrow i} (\tilde{\mathbf{x}}_j^k-\tilde{\mathbf{x}}_i^k)] = w_{i \rightarrow j} w_{j \rightarrow i} (\mathbf{x}_j^k-\mathbf{x}_i^k)$, we have $\mathbb{E}[\boldsymbol{\xi}_i^k] = \mathbf{0}_d$. Let $\mathbf{x}_i^k = [x_{i1}^k, \ldots, x_{id}^k]^T$. Since $\tilde{w}_{i \rightarrow j}=w_{i \rightarrow j} \in (0, \, 1]$, for any $l \in {1, \ldots, d}$, from \eqref{eq_inverse_quan} we have
\begin{equation*}
	\begin{aligned}
		 & \mathbb{E}\big[ \big(\tilde{w}_{i \rightarrow j} \tilde{w}_{j \rightarrow i} (\tilde{x}_{jl}^k - \tilde{x}_{il}^k) - w_{i \rightarrow j} w_{j \rightarrow i} (x_{jl}^k - x_{il}^k) \big)^2 \big]\\
		\leq & \mathbb{E}\big[\big(\tilde{x}_{jl}^k - x_{j l}^k)^2\big]+\mathbb{E}\big[\big(\tilde{x}_{il}^k-x_{i l}^k)^2\big]\leq\frac{1}{2}\delta^2,
	\end{aligned}
\end{equation*}
where we used the fact that the quantization of $x_{il}^k$ is independent of $x_{jl}^k$. Therefore, $\mathbb{E}\big[ \| \boldsymbol{\xi}_i^k \|^2 \big]$ is bounded by:
\begin{equation*}\label{Exi}
	\begin{aligned}
		  \mathbb{E}[\|\boldsymbol{\xi}_i^k\|^2] 
		=&  \mathbb{E}[\| \smallsum\nolimits_{j \in \mathcal{N}_i} ( w_{i \rightarrow j} w_{j \rightarrow i} (\tilde{\mathbf{x}}_j^k - \tilde{\mathbf{x}}_i^k - \mathbf{x}_j^k + \mathbf{x}_i^k)) \|^2] \\
		\leq&  \mathbb{E} [ \mathcal{N}_i^2 \smallsum\nolimits_{l=1}^d ( w_{i \rightarrow j} w_{j \rightarrow i} (\tilde{x}_{j l}^k - \tilde{x}_{i l}^k - x_{j l}^k + x_{i l}^k) )^2 ] \\
		\leq& \frac{1}{2}(n-1)^2d\delta^2.
	\end{aligned}
\end{equation*}
The network coupling is characterized by a symmetric weight matrix $W=[w_{ij}] \in \mathbb{R}^{n\times n}$:
\begin{equation}\label{eq:W}
	w_{ij} = 
	\begin{cases}
		w_{i \rightarrow j} w_{j \rightarrow i}, & \text{if } j \in \mathcal{N}_i, \\
		- \sum\nolimits_{l \in \mathcal{N}_i} w_{i \rightarrow l} w_{l \rightarrow i}, & \text{if } i = j, \\
		0, & \text{if } j \notin \mathcal{N}_i \cup \{i\}.		
	\end{cases}
\end{equation}

By defining the concatenated vectors and matrices:
\begin{equation*}
	\begin{aligned}
		\Lambda^k & =\operatorname{diag}\{\Lambda_1^k, \ldots, \Lambda_n^k\} \in \mathbb{R}^{n d \times n d}, \\
		\mathbf{x}^{k} & =\big[(\mathbf{x}_1^k)^T, \ldots, (\mathbf{x}_n^k)^T\big]^T \in \mathbb{R}^{n d \times 1}, \\
		\boldsymbol{\xi}^k & =\big[(\boldsymbol{\xi}_1^k)^T, \ldots, (\boldsymbol{\xi}_n^k)^T\big]^T \in \mathbb{R}^{n d \times 1}, \\
		\mathbf{g}^k & =\big[(\mathbf{g}_1^k)^T, \ldots, (\mathbf{g}_n^k)^T\big]^T \in \mathbb{R}^{n d \times 1}, 
	\end{aligned}
\end{equation*}
we obtain the compact form of the algorithm:
\begin{equation}\label{bfx}
	\mathbf{x}^{k+1}=(I_{n d}+\gamma^k W \otimes I_d) \mathbf{x}^k+\gamma^k \boldsymbol{\xi}^k-\Lambda^k \mathbf{g}^k.
\end{equation}

\begin{remark}
	As shown in \eqref{bfx}, the attenuation factor $\gamma^k$ serves to suppress the quantization error induced by the encryption process, thereby maintaining algorithm convergence throughout the optimization process.
\end{remark}

\subsection{Convergence Analysis}
In order to demonstrate convergence of the algorithm, we begin by presenting several preliminary lemmas.

\begin{lemma}\label{lemma:W}
	Under Assumption \ref{ass:graph}, the eigenvalues of the symmetric coupling weight matrix $W$ defined in \eqref{eq:W} can be ordered as $0 = \rho_1 > \rho_2 \geq \cdots \geq \rho_n$. Define $\bar{W}^k = I_n + \gamma^k W - \frac{1}{n} \mathbf{1}_n \mathbf{1}_n^T$, where $\gamma^k > 0$ is the attenuation factor. Then, there is some positive integer $K$ such that for all $k \ge K$ satisfying $\gamma^k \le 1 / |\rho_n|$, the spectral norm obeys
	\begin{equation}\label{eq:barW}
		\|\bar{W}^k\| = 1 - \mu\gamma^k \in [0, \, 1), 
	\end{equation}
	where $\mu = |\rho_2| > 0$.
\end{lemma}

\begin{pf}
	Since $W$ is symmetric and satisfies $W \mathbf{1}_n = \mathbf{0}_n$, it has a zero eigenvalue $\rho_1 = 0$ with eigenvector $\mathbf{1}_n$. Given that $W$ corresponds to a connected graph, its zero eigenvalue is simple, and the remaining eigenvalues $\rho_2, \ldots, \rho_n$ are strictly negative. Therefore, $\mu = |\rho_2| = \min_{i \geq 2} |\rho_i| > 0$. The eigenvalues of the symmetric matrix $\bar{W}^k$ can be characterized as follows:
	
	(1) For eigenvector $\mathbf{1}_n$:
	\begin{equation*}
	\bar{W}^k \mathbf{1}_n = \big(I + \gamma^k W - \frac{1}{n} \mathbf{1}_n \mathbf{1}_n^T \big) \mathbf{1}_n = \mathbf{1}_n + \mathbf{0}_n - \mathbf{1}_n = 0 \cdot \mathbf{1}_n, 
	\end{equation*}
	so the corresponding eigenvalue is $0$.
	
	(2) For any eigenvector $\mathbf{u}_i$ of $W$ orthogonal to $\mathbf{1}_n$ (i.e., $W \mathbf{u}_i = \rho_i \mathbf{u}_i$ for $i \geq 2$):
	\begin{equation*}
	\bar{W}^k \mathbf{u}_i = \big(I + \gamma^k W - \frac{1}{n} \mathbf{1}_n \mathbf{1}_n^T \big) \mathbf{u}_i = (1 + \gamma^k \rho_i) \mathbf{u}_i, 
	\end{equation*}
	so the corresponding eigenvalue is $1 + \gamma^k \rho_i$ for $i=2,\ldots,n$.
	
	Since $\rho_i < 0$ for $i \geq 2$ and $\gamma^k > 0$, we have $1 + \gamma^k \rho_i < 1$. Choosing $K$ sufficiently large so that $\gamma^k \leq 1/|\rho_n|$ for all $k \geq K$ ensures $1 + \gamma^k \rho_i \geq 0$ for $i \geq 2$. Thus, the spectral norm of $\bar{W}^k$ for $k \geq K$ is given by the largest eigenvalue:
	\begin{equation*}
		\|\bar{W}^k\| = \max \big\{ 0, 1 + \gamma^k \rho_i \, | \, i \geq 2 \big\} = 1 + \gamma^k \rho_2, 
	\end{equation*}
 which concludes the proof. \hfill$\blacksquare$
\end{pf}

Next, we derive upper bounds for $\mathbb{E}[\|\bar{\mathbf{x}}^{k+1}-\mathbf{x}^*\|^2 \, | \, \mathcal{F}^k]$ and $\mathbb{E}[\|\mathbf{x}^{k+1}-\mathbf{1}_n \otimes \bar{\mathbf{x}}^{k+1}\|^2 \, | \, \mathcal{F}^k]$, where $\mathcal{F}^k$ denotes the $\sigma$-algebra generated by $\{\boldsymbol{\xi}^0, \mathbf{g}^0, \ldots, \boldsymbol{\xi}^{k-1}, \mathbf{g}^{k-1}\}$.

\begin{lemma}\label{lemma:first}
	Under Assumptions \ref{ass:Ff} and \ref{ass:step}, the following inequality holds for $k \geq 0$:
		\begin{align}
			 & \mathbb{E}[\|\bar{\mathbf{x}}^{k+1}-\mathbf{x}^*\|^2 \, | \, \mathcal{F}^k] \nonumber\\
			\leq & \Big(1+\frac{nL^2(\bar{\lambda}^k)^2}{\gamma^k}+\frac{\sqrt{d}(2 nL^2+1)}{n}\|\boldsymbol{\lambda}^k-\bar{\lambda}^k \mathbf{1}_n\|\nonumber\\
			 & +\frac{8dL^2\sum\nolimits_{i=1}^n \theta_i^k}{n}\Big)\|\bar{\mathbf{x}}^k-\mathbf{x}^*\|^2\nonumber\\ 
			 & +\Big(\frac{\gamma^k}{n}+\frac{8dL^2\sum\nolimits_{i=1}^n \theta_i^k}{n^2}\Big)\|\mathbf{x}^{k}-\mathbf{1}_n \otimes \bar{\mathbf{x}}^{k}\|^2\nonumber\\
			 & +\Big(\frac{4d\sum\nolimits_{i=1}^n \theta_i^k}{n^2}+\frac{2\sqrt{d}}{n}\|\boldsymbol{\lambda}^k-\bar{\lambda}^k\mathbf{1}_n\|\Big)\sum\nolimits_{i=1}^n\|\nabla f_i(\mathbf{x}^*)\|^2\nonumber\\
			 & +2(n-1)^2d\delta^2(\gamma^k)^2+\frac{2d\sum\nolimits_{i=1}^n\theta_i^k}{n}\sigma_g^2\nonumber\\
			 & -2\bar{\lambda}^k\big(F(\bar{\mathbf{x}}^k)-F(\mathbf{x}^*)\big), \label{part1}
		\end{align}
	where $\bar{\lambda}^k = \frac{1}{n}\sum\nolimits_{i=1}^n \lambda_i^k$, $\boldsymbol{\lambda}^k=[\lambda_1^k, \ldots, \lambda_n^k]^T$, $\theta_i^k = (\lambda_i^k)^2 + (\sigma_i^k)^2$, and $\mathbf{x}^*$ is an optimal solution to \eqref{Fd}. 
\end{lemma}
\begin{pf}
	Given $\bar{\mathbf{x}}^k=\frac{1}{n} \sum\nolimits_{i=1}^n \mathbf{x}_i^k$ and combined with \eqref{bfx}, 
	we obtain
	\begin{equation}\label{barx}
		\begin{aligned}
			\bar{\mathbf{x}}^{k+1}= & \frac{1}{n}\mathbf{1}_n^T \otimes I_d \, \mathbf{x}^{k+1}\\ 
			= & \frac{1}{n}\mathbf{1}_n^T \otimes I_d (I_{nd}+\gamma^k W \otimes I_d) \mathbf{x}^k\\
			 & +\frac{1}{n}\mathbf{1}_n^T \otimes I_d\gamma^k\boldsymbol{\xi}^k-\frac{1}{n}\mathbf{1}_n^T \otimes I_d \Lambda^k \mathbf{g}^k \\
			= & \bar{\mathbf{x}}^k+\frac{\gamma^k}{n}\sum\nolimits_{i=1}^n\boldsymbol{\xi}_i^k-\frac{1}{n}\sum\nolimits_{i=1}^n \Lambda_i^k \mathbf{g}_i^k, 
		\end{aligned}
	\end{equation}
	where in the derivation we used $\mathbf{1}_n^T W=\mathbf{0}^T$. Thus, 
	\begin{equation}
		\begin{aligned}
			 & \|\bar{\mathbf{x}}^{k+1}-\mathbf{x}^*\|^2\\
			= & \big\|\bar{\mathbf{x}}^k-\mathbf{x}^*+\frac{\gamma^k}{n}\sum\nolimits_{i=1}^n\boldsymbol{\xi}_i^k-\frac{1}{n} \sum\nolimits_{i=1}^n \Lambda_i^k \mathbf{g}_i^k\big\|^2\\
			= & \|\bar{\mathbf{x}}^k-\mathbf{x}^*\|^2+\frac{(\gamma^k)^2}{n^2}\big\|\sum\nolimits_{i=1}^n \boldsymbol{\xi}_i^k \big\|^2+\frac{1}{n^2}\big\|\sum\nolimits_{i=1}^n \Lambda_i^k \mathbf{g}_i^k\big\|^2\\
			 & + \frac{2\gamma^k}{n}\big\langle \bar{\mathbf{x}}^k-\mathbf{x}^*, \sum\nolimits_{i=1}^n\boldsymbol{\xi}_i^k \big\rangle-\frac{2}{n} \big\langle \bar{\mathbf{x}}^k-\mathbf{x}^*, \sum\nolimits_{i=1}^n \Lambda_i^k \mathbf{g}_i^k \big\rangle\\
			 & -\frac{2\gamma^k}{n^2}\big\langle\sum\nolimits_{i=1}^n\boldsymbol{\xi}_i^k, \sum\nolimits_{i=1}^n \Lambda_i^k \mathbf{g}_i^k\big\rangle.
		\end{aligned}
	\end{equation}
	Given $\mathbb{E}[\boldsymbol{\xi}_i^k] = \mathbf{0}_d$ and $\mathbb{E}[\|\boldsymbol{\xi}_i^k\|^2] \leq \frac{1}{2}d(n-1)^2\delta^2$, it follows that
	\begin{equation}\label{Ebar}
		\begin{aligned}
			 & \mathbb{E}[\|\bar{\mathbf{x}}^{k+1}-\mathbf{x}^*\|^2 \, | \, \mathcal{F}^k]\\
			\leq & \|\bar{\mathbf{x}}^k-\mathbf{x}^*\|^2+\frac{1}{2}(n-1)^2d\delta^2(\gamma^k)^2\\
			 & +\frac{1}{n^2}\mathbb{E}\big[\|\sum\nolimits_{i=1}^n \Lambda_i^k \mathbf{g}_i^k\|^2 \, | \, \mathcal{F}^k\big]\\
			 & -\frac{2}{n}\mathbb{E}\big[\langle \bar{\mathbf{x}}^k-\mathbf{x}^*, \sum\nolimits_{i=1}^n \Lambda_i^k \mathbf{g}_i^k \rangle \, | \, \mathcal{F}^k\big].
		\end{aligned}
	\end{equation}
	
	Define $\hat{\Lambda }^k=[\Lambda_1^k, \ldots, \Lambda_n^k] \in \mathbb{R}^{d \times nd}$, then we can deduce 
	\begin{equation}\label{lambdag}
		\begin{aligned}
			 & \frac{1}{n^2}\mathbb{E}\big[\|\sum\nolimits_{i=1}^n \Lambda_i^k \mathbf{g}_i^k\|^2 \, | \, \mathcal{F}^k\big] \\
			= & \frac{1}{n^2}\mathbb{E}\big[\|\hat{\Lambda }^k \mathbf{g}^k\|^2 \, | \, \mathcal{F}^k\big]\\
			\leq & \frac{1}{n^2} \mathbb{E}\big[\|\hat{\Lambda }^k \|^2 \, | \, \mathcal{F}^k\big] \mathbb{E}\big[ \|\mathbf{g}^k\|^2 \, | \, \mathcal{F}^k\big].
		\end{aligned}
	\end{equation}
	Given $\|\hat{\Lambda }^k\|^2 \leq\|\hat{\Lambda }^k\|_F^2=\sum\nolimits_{i=1}^n \sum\nolimits_{l=1}^d(\lambda_{il}^k)^2$, $\mathbb{E}[\lambda_{il}^k]=\lambda_i^k$, and $\mathbb{D} [\lambda_{il}^k]=(\sigma_i^k)^2$, we can obtain
	\begin{equation}\label{hatlambda1}
		\mathbb{E}\big[\|\hat{\Lambda }^k \|^2 \, | \, \mathcal{F}^k\big]\leq\sum\nolimits_{i=1}^n d((\lambda_i^k)^2 + (\sigma_i^k)^2)=\sum\nolimits_{i=1}^n d\theta_i^k.
	\end{equation}
	
	Let $\nabla f(\mathbf{x}^k)=[(\nabla f_1(\mathbf{x}_1^k))^T, \ldots, (\nabla f_n(\mathbf{x}_n^k))^T]^T$, then
	\begin{equation}\label{g}
		\begin{aligned}
			\mathbb{E}\big[\|\mathbf{g}^k\|^2 \, | \, \mathcal{F}^k\big]= & \mathbb{E}\big[\|\mathbf{g}^k-\nabla f(\mathbf{x}^k)+\nabla f(\mathbf{x}^k)\|^2 \, | \, \mathcal{F}^k\big] \\
			\leq & \mathbb{E}\big[2\|\mathbf{g}^k-\nabla f(\mathbf{x}^k)\|^2 \, | \, \mathcal{F}^k\big]+2\|\nabla f(\mathbf{x}^k)\|^2\\
			\leq & 2n\sigma_g^2+2\|\nabla f(\mathbf{x}^k)\|^2, 
		\end{aligned}
	\end{equation}
	where we used $\mathbb{E}_{\boldsymbol{\eta}_i \sim \mathcal{O}_i}[\|\mathbf{g}_i^k-\nabla f_i(\mathbf{x}_i^k)\|^2] \leq \sigma_g^2$.
	
	For $\|\nabla f(\mathbf{x}^k)\|^2$ in \eqref{g}, we have
	\begin{equation}
		\begin{aligned}\label{f}
			 & \|\nabla f(\mathbf{x}^k)\|^2\\
			= & \|\nabla f(\mathbf{x}^k)-\nabla f(\mathbf{1}_n\otimes\mathbf{x}^*)+\nabla f(\mathbf{1}_n\otimes\mathbf{x}^*)\|^2 \\
			\leq & 2\|\nabla f(\mathbf{x}^k)-\nabla f(\mathbf{1}_n\otimes\mathbf{x}^*)\|^2+2\|\nabla f(\mathbf{1}_n\otimes\mathbf{x}^*)\|^2 \\
			\leq & 2 L^2\|\mathbf{x}^k-\mathbf{1}_n\otimes\mathbf{x}^*\|^2+2\sum\nolimits_{i=1}^n\|\nabla f_i(\mathbf{x}^*)\|^2 \\
			\leq & 2L^2\|\mathbf{x}^k-\mathbf{1}_n\otimes\bar{\mathbf{x}}^k+\mathbf{1}_n\otimes\bar{\mathbf{x}}^k-\mathbf{1}_n\otimes\mathbf{x}^*\|^2\\
			 & +2\sum\nolimits_{i=1}^n\|\nabla f_i(\mathbf{x}^*)\|^2\\
			\leq & 4 L^2 \sum\nolimits_{i=1}^n\|\mathbf{x}_i^k-\mathbf{1}_n\otimes\bar{\mathbf{x}}^k\|^2+4 n L^2\|\bar{\mathbf{x}}^k-\mathbf{x}^*\|^2\\
			 & +2\sum\nolimits_{i=1}^n\|\nabla f_i(\mathbf{x}^*)\|^2, 
		\end{aligned}
	\end{equation}
	where in the derivation we use the Lipschitz continuity of $\nabla f_i$ and $(a+b)^2 \leq 2a^2+2b^2$.
	
	Substituting \eqref{hatlambda1}, \eqref{g}, and \eqref{f} into \eqref{lambdag} yields
	\begin{equation}\label{Exbar}
		\begin{aligned}
			 & \frac{1}{n^2}\mathbb{E}\big[\|\sum\nolimits_{i=1}^n \Lambda_i^k \mathbf{g}_i^k\|^2 \, | \, \mathcal{F}^k\big]\\
			\leq & \frac{d}{n^2}\sum\nolimits_{i=1}^n \theta_i^k\big[8 n L^2\|\bar{\mathbf{x}}^k-\mathbf{x}^*\|^2+8 L^2 \|\mathbf{x}^k-\mathbf{1}_n\otimes\bar{\mathbf{x}}^k\|^2\\
			 & +4\sum\nolimits_{i=1}^n\|\nabla f_i(\mathbf{x}^*)\|^2+2 n \sigma_g^2\big].
		\end{aligned}
	\end{equation}
	
	Next, we analyze the last term of \eqref{Ebar}:
	\begin{equation}\label{lamdaf}
		\begin{aligned}
			 & -\frac{2}{n} \mathbb{E}\big[\langle \bar{\mathbf{x}}^k-\mathbf{x}^*, \sum\nolimits_{i=1}^n \Lambda_i^k \mathbf{g}_i^k\rangle \, | \, \mathcal{F}^k\big] \\
			= & -\frac{2}{n}\sum\nolimits_{i=1}^{n} \langle\lambda_i^k \nabla f_i(\mathbf{x}_i^k), \bar{\mathbf{x}}^k-\mathbf{x}^*\rangle\\
			= & \frac{2}{n}\sum\nolimits_{i=1}^{n}\langle\lambda_i^k(\nabla f_i(\bar{\mathbf{x}}^k)-\nabla f_i(\mathbf{x}_i^k)), \bar{\mathbf{x}}^k-\mathbf{x}^*\rangle \\
			 & -\frac{2}{n}\sum\nolimits_{i=1}^{n}\langle\lambda_i^k \nabla f_i(\bar{\mathbf{x}}^k), \bar{\mathbf{x}}^k-\mathbf{x}^*\rangle.
		\end{aligned}
	\end{equation}
	
	For the first term in the right hand side of \eqref{lamdaf}, we have
	\begin{equation}\label{lambdaf-f}
		\begin{aligned}
			 & \frac{2}{n}\sum\nolimits_{i=1}^{n}\langle\lambda_i^k(\nabla f_i(\bar{\mathbf{x}}^k)-\nabla f_i(\mathbf{x}_i^k)), \bar{\mathbf{x}}^k-\mathbf{x}^*\rangle \\
			\leq & \frac{2}{n}\sum\nolimits_{i=1}^{n}\lambda_i^kL\|\mathbf{x}_i^k-\bar{\mathbf{x}}^k\|\cdot\|\bar{\mathbf{x}}^k-\mathbf{x}^*\| \\
			\leq & \frac{\gamma^k}{n}\sum\nolimits_{i=1}^{n}\|\mathbf{x}_i^k-\bar{\mathbf{x}}^k\|^2+\frac{L^2}{n\gamma^k}\sum\nolimits_{i=1}^{n}(\lambda_i^k)^2\|\bar{\mathbf{x}}^k-\mathbf{x}^*\|^2\\
			\leq & \frac{\gamma^k}{n}\|\mathbf{x}^k-\mathbf{1}_n\otimes\bar{\mathbf{x}}^k\|^2+\frac{nL^2(\bar{\lambda}^k)^2}{\gamma^k}\sum\nolimits_{i=1}^{n}\|\bar{\mathbf{x}}^k-\mathbf{x}^*\|^2, 
		\end{aligned}
	\end{equation}
	where in the derivation we used $\sum\nolimits_{i=1}^n(\lambda_i^k)^2\leq n^2(\bar{\lambda}^k)^2$ with $\bar{\lambda}^k=\frac{1}{n}\sum\nolimits_{i=1}^n\lambda_i^k$.
	
	For the second term in the right hand side of \eqref{lamdaf}, we have
	\begin{equation}\label{lambdaxbar1}
		\begin{aligned}
			 & -\frac{2}{n}\sum\nolimits_{i=1}^{n}\langle\lambda_i^k \nabla f_i(\bar{\mathbf{x}}^k), \bar{\mathbf{x}}^k-\mathbf{x}^*\rangle\\
			= & \frac{2}{n}\sum\nolimits_{i=1}^n\langle(\lambda_i^k-\bar{\lambda}^k) \nabla f_i(\bar{\mathbf{x}}^k), \mathbf{x}^*-\bar{\mathbf{x}}^k\rangle\\
			 & +\frac{2}{n}\sum\nolimits_{i=1}^n\langle\bar{\lambda}^k \nabla f_i(\bar{\mathbf{x}}^k), \mathbf{x}^*-\bar{\mathbf{x}}^k\rangle.
		\end{aligned}
	\end{equation}
	Regarding the last term in \eqref{lambdaxbar1}, we obtain
	\begin{equation}\label{lambdaxbar}
		\begin{aligned}
			 & \frac{2}{n}\sum\nolimits_{i=1}^n\langle\bar{\lambda}^k \nabla f_i(\bar{\mathbf{x}}^k), \mathbf{x}^*-\bar{\mathbf{x}}^k\rangle\\
			= & 2\bar{\lambda}^k\langle \nabla F(\bar{\mathbf{x}}^k), \mathbf{x}^*-\bar{\mathbf{x}}^k\rangle
			\leq-2\bar{\lambda}^k(F(\bar{\mathbf{x}}^k)-F(\mathbf{x}^*)), 
		\end{aligned}
	\end{equation}
	where we used the convexity of $F(\cdot)$. Then we analyze the first term in the right hand side of \eqref{lambdaxbar1}. Given $\boldsymbol{\lambda}^k=[\lambda_1^k, \ldots, \lambda_n^k]^T$, we have
	\begin{equation}\label{innerlambda}
		\begin{aligned}
			 & \frac{2}{n}\sum\nolimits_{i=1}^n\langle(\lambda_i^k-\bar{\lambda}^k) \nabla f_i(\bar{\mathbf{x}}^k), \mathbf{x}^*-\bar{\mathbf{x}}^k\rangle\\ 
			= & \frac{2}{n}\langle\sum\nolimits_{i=1}^n(\lambda_i^k-\bar{\lambda}^k) \nabla f_i(\bar{\mathbf{x}}^k), \mathbf{x}^*-\bar{\mathbf{x}}^k\rangle \\
			\leq & \frac{2}{n}\|\sum\nolimits_{i=1}^n(\lambda_i^k-\bar{\lambda}^k) \nabla f_i(\bar{\mathbf{x}}^k)\|\cdot\|\bar{\mathbf{x}}^k-\mathbf{x}^*\| \\
			= & \frac{2}{n}\|((\boldsymbol{\lambda}^k-\bar{\lambda}^k \mathbf{1}_n)^T \otimes I_d) \nabla f(\mathbf{1}_n \otimes \bar{\mathbf{x}}^k)\|\cdot\|\bar{\mathbf{x}}^k-\mathbf{x}^*\| \\
			\leq & \frac{\sqrt{d}}{n}\|\boldsymbol{\lambda}^k-\bar{\lambda}^k \mathbf{1}_n\| \big(\|\nabla f(\mathbf{1}_n \otimes \bar{\mathbf{x}}^k)\|^2+\|\bar{\mathbf{x}}^k-\mathbf{x}^*\|^2 \big).
		\end{aligned}
	\end{equation}
	Moreover, $\|\nabla f(\mathbf{1}_n \otimes \bar{\mathbf{x}}^k)\|^2$ in \eqref{innerlambda} satisfies
	\begin{equation}\label{fn}
		\begin{aligned}
			 & \|\nabla f(\mathbf{1}_n \otimes \bar{\mathbf{x}}^k)\|^2 \\
			= & \|\nabla f(\mathbf{1}_n \otimes \bar{\mathbf{x}}^k)-\nabla f(\mathbf{1}_n\otimes\mathbf{x}^*)+\nabla f(\mathbf{1}_n\otimes\mathbf{x}^*)\|^2 \\
			\leq & 2\|\nabla f(\mathbf{1}_n \otimes \bar{\mathbf{x}}^k)-\nabla f(\mathbf{1}_n\otimes\mathbf{x}^*)\|^2 + 2\|\nabla f(\mathbf{1}_n\otimes\mathbf{x}^*)\|^2 \\
			= & 2nL^2\|\bar{\mathbf{x}}^k-\mathbf{x}^*\|^2+2\sum\nolimits_{i=1}^n\|\nabla f_i(\mathbf{x}^*)\|^2.
		\end{aligned}
	\end{equation}
	Therefore, combining \eqref{lambdaxbar1}, \eqref{lambdaxbar}, \eqref{innerlambda}, and \eqref{fn}, we obtain
	\begin{equation}\label{lambdaf2}
		\begin{aligned}
			 & -\frac{2}{n}\sum\nolimits_{i=1}^{n}\langle\lambda_i^k \nabla f_i(\bar{\mathbf{x}}^k), \bar{\mathbf{x}}^k-\mathbf{x}^*\rangle\\
			\leq & \frac{\sqrt{d}(2 nL^2+1)}{n}\|\boldsymbol{\lambda}^k-\bar{\lambda}^k \mathbf{1}_n\|\cdot\|\bar{\mathbf{x}}^k-\mathbf{x}^*\|^2\\
			 & +\frac{2\sqrt{d}}{n}\|\boldsymbol{\lambda}^k-\bar{\lambda}^k\mathbf{1}_n\| \sum\nolimits_{i=1}^n\|\nabla f_i(\mathbf{x}^*)\|^2\\
			 & -2\bar{\lambda}^k\big(F(\bar{\mathbf{x}}^k)-F(\mathbf{x}^*)\big).
		\end{aligned}
	\end{equation}
	Combining \eqref{lambdaf-f} and \eqref{lambdaf2} yields
	\begin{equation}\label{nablaf}
		\begin{aligned}
			 & -\frac{2}{n}\mathbb{E}\big[\langle \bar{\mathbf{x}}^k-\mathbf{x}^*, \sum\nolimits_{i=1}^n \Lambda_i^k g_i^k \rangle \, | \, \mathcal{F}^k\big] \\
			\leq & \Big(\frac{nL^2(\bar{\lambda}^k)^2}{\gamma^k}+\frac{\sqrt{d}(2 nL^2+1)}{n}\|\boldsymbol{\lambda}^k-\bar{\lambda}^k \mathbf{1}_n\|\Big)\|\bar{\mathbf{x}}^k-\mathbf{x}^*\|^2\\
			 & +\frac{\gamma^k}{n}\|\mathbf{x}^k-\mathbf{1}_n\otimes\bar{\mathbf{x}}^k\|^2-2\bar{\lambda}^k\big(F(\bar{\mathbf{x}}^k)-F(\mathbf{x}^*)\big)\\
			 & +\frac{2\sqrt{d}}{n}\|\boldsymbol{\lambda}^k-\bar{\lambda}^k\mathbf{1}_n\| \sum\nolimits_{i=1}^n\|\nabla f_i(\mathbf{x}^*)\|^2.
		\end{aligned}
	\end{equation}
	
	Finally, plugging \eqref{Exbar} and \eqref{nablaf} into \eqref{Ebar}, we have \eqref{part1}. \hfill$\blacksquare$
\end{pf}

\begin{lemma}\label{lemma2}
	Under Assumptions \ref{ass:graph}, \ref{ass:Ff}, and \ref{ass:step}, there exists $K\in \mathbb{Z}^+$ such that for all $k \geq K$, the following inequality holds:
	\begin{equation}\label{part2}
		\begin{aligned}
			 & \mathbb{E}[\|\mathbf{x}^{k+1}-\mathbf{1}_n \otimes \bar{\mathbf{x}}^{k+1}\|^2 \, | \, \mathcal{F}^k]\\ 
			\leq & \Big(1 - \mu\gamma^k + \frac{8L^2\theta_{\max}^k}{\mu\gamma^k} \Big) \|\mathbf{x}^{k}-\mathbf{1}_n \otimes \bar{\mathbf{x}}^{k}\|^2\\
			 & +\frac{8 n L^2\theta_{\max}^k}{\mu\gamma^k }\|\bar{\mathbf{x}}^k-\mathbf{x}^*\|^2+\frac{2n\theta_{\max}^k}{\mu\gamma^k }\sigma_g^2\\
			 & +\frac{4\theta_{\max}^k}{\mu\gamma^k }\sum\nolimits_{i=1}^n\|\nabla f_i(\mathbf{x}^*)\|^2+\frac{1}{2}(n-1)^2d\delta^2(\gamma^k)^2, 
		\end{aligned}
	\end{equation}
	where $\mu = |\rho_2|$ is the positive constant given by the minimum absolute nonzero eigenvalue of $W$ and $\theta_{\max}^k = \max\{ \theta_i^k \, | \, 1 \leq i \leq n \} = \max\{ (\lambda_i^k)^2+(\sigma_i^k)^2 \, | \, 1 \leq i \leq n \}$.
\end{lemma}
\begin{pf}
	Given $\bar{\mathbf{x}}^{k}=\frac{1}{n}(\mathbf{1}_n^T\otimes I_d)\mathbf{x}^k$, from \eqref{bfx} and \eqref{barx} we have
	\begin{equation}\label{mathbfx1}
		\begin{aligned}
			& \mathbf{x}^{k+1}-\mathbf{1}_n \otimes \bar{\mathbf{x}}^{k+1} \\
			= & (I_{nd}-\frac{1}{n}\mathbf{1}_n\mathbf{1}_n^T\otimes I_d)\mathbf{x}^{k+1}\\
			= & (I_{nd} + \gamma^kW \otimes I_d-\frac{1}{n}\mathbf{1}_n\mathbf{1}_n^T\otimes I_d)\mathbf{x}^{k}\\
			& + (I_{nd}-\frac{1}{n}\mathbf{1}_n\mathbf{1}_n^T\otimes I_d) (\gamma^k\boldsymbol{\xi}^k - \Lambda^k\mathbf{g}^k) \\
			= & (\bar{W}^k\otimes I_d) \mathbf{x}^k+\gamma^k R\boldsymbol{\xi}^k-R\Lambda^k \mathbf{g}^k, 
		\end{aligned}	 	
	\end{equation}
	where $\bar{W}^k = I_n + \gamma^kW- \frac{1}{n}\mathbf{1}_n\mathbf{1}_n^T$ and $R = (I_n - \frac{1}{n} \mathbf{1}_n \mathbf{1}_n^T ) \otimes I_d$. Since $\bar{W} \mathbf{1}_n = \mathbf{0}_n$, it follows that
	\begin{equation}
		\begin{aligned}
			 & (\bar{W}^k\otimes I_d)(\mathbf{1}_n \otimes \bar{\mathbf{x}}^k)=\big(\bar{W}^k\times\mathbf{1}_n\big) \otimes(I_d \times \bar{\mathbf{x}}^k)=0, 
		\end{aligned}
	\end{equation}
	where we used the property $(A \otimes B)(C \otimes D)=(A C) \otimes(B D)$ in the derivation. Therefore, \eqref{mathbfx1} can be rewritten as
	\begin{equation}\label{mathbfx2}
		\begin{aligned}
			& \mathbf{x}^{k+1}-\mathbf{1}_n \otimes \bar{\mathbf{x}}^{k+1} \\ 
			= & (\bar{W}^k\otimes I_d)(\mathbf{x}^k-\mathbf{1}_n \otimes \bar{\mathbf{x}}^{k}) +\gamma^k R\boldsymbol{\xi}^k-R\Lambda^k \mathbf{g}^k.
		\end{aligned}
	\end{equation}
	Taking the squared norm on both sides of \eqref{mathbfx2} yields
	\begin{equation}\label{normx}
		\begin{aligned}
			 & \|\mathbf{x}^{k+1}-\mathbf{1}_n \otimes \bar{\mathbf{x}}^{k+1}\|^2\\
			= & \|(\bar{W}^k\otimes I_d)(\mathbf{x}^k-\mathbf{1}_n \otimes \bar{\mathbf{x}}^{k})-R\Lambda^k \mathbf{g}^k\|^2 +\|\gamma^k R\boldsymbol{\xi}^k\|^2\\
			 & +2\langle (\bar{W}^k\otimes I_d)(\mathbf{x}^k-\mathbf{1}_n \otimes \bar{\mathbf{x}}^{k})-R\Lambda^k \mathbf{g}^k, \gamma^k R\boldsymbol{\xi}^k\rangle\\
			\leq & \|(\bar{W}^k\otimes I_d)(\mathbf{x}^k-\mathbf{1}_n \otimes \bar{\mathbf{x}}^{k})-R\Lambda^k \mathbf{g}^k\|^2+(\gamma^k)^2\| \boldsymbol{\xi}^k\|^2\\
			 & +2\langle(\bar{W}^k\otimes I_d)(\mathbf{x}^k-\mathbf{1}_n \otimes \bar{\mathbf{x}}^{k})-R\Lambda^k \mathbf{g}^k, \gamma^k R\boldsymbol{\xi}^k\rangle,
		\end{aligned}
	\end{equation}
	where we used $\|R\|=1$ since the eigenvalues of $R$ are either $1$ or $0$.
	
	Given the $\sigma$-algebra $\mathcal{F}^k$, $\boldsymbol{\xi}^k$ is independent of $\mathbf{x}^k$ and $\Lambda^k\mathbf{g}^k$. Taking the conditional expectation and using $\mathbb{E}[\boldsymbol{\xi}_i^k]= \mathbf{0}_d$ and $\mathbb{E}[\|\boldsymbol{\xi}_i^k\|^2] \leq \frac{1}{2}(n-1)^2d\delta^2$, we obtain
	\begin{equation}\label{Emathbfx}
		\begin{aligned}
			 & \mathbb{E}\left[\|\mathbf{x}^{k+1}-\mathbf{1}_n \otimes \bar{\mathbf{x}}^{k+1}\|^2 \, | \, \mathcal{F}^k\right]\\ 
			\leq & \mathbb{E}\big[\|(\bar{W}^k\otimes I_d)(\mathbf{x}^k-\mathbf{1}_n \otimes \bar{\mathbf{x}}^{k})-R\Lambda^k\mathbf{g}^k\|^2 \, | \, \mathcal{F}^k\big]\\
			 & + \frac{1}{2}(n-1)^2d\delta^2(\gamma^k)^2\\
			\leq & \mathbb{E}\big[ \big(\|\bar{W}^k\| \cdot \|\mathbf{x}^k-\mathbf{1}_n \otimes \bar{\mathbf{x}}^{k}\| + \|\Lambda^k\| \cdot \| \mathbf{g}^k\|\big)^2 \, | \, \mathcal{F}^k\big]\\
			 & + \frac{1}{2}(n-1)^2d\delta^2(\gamma^k)^2, 
		\end{aligned}
	\end{equation}
	where in the derivation we used $\|\bar{W}^k\otimes I_d\|=\|\bar{W}^k\|$ and $\|R\|=1$. 
	
	From Lemma \ref{lemma:W}, there exists $K\in \mathbb{Z}^+$ such that for all $k \geq K$, we have $\|\bar{W}^k\| = 1 - \mu\gamma^k\in(0, 1)$. Applying the inequality $(a+b)^2 \leq (1+\alpha)a^2 + (1+\alpha^{-1})b^2$ for any $a, b \in \mathbb{R}$ and $\alpha > 0$, and setting $\alpha=\frac{\mu\gamma^k}{1-\mu\gamma^k}$, for $k \geq K$ we have
	\begin{equation}\label{Emathbfx1}
		\begin{aligned}
			 & \mathbb{E}\big[\big(\|\bar{W}^k\|\cdot\|\mathbf{x}^k-\mathbf{1}_n \otimes \bar{\mathbf{x}}^{k}\| + \|\Lambda^k\|\cdot\| \mathbf{g}^k\|\big)^2 \, | \, \mathcal{F}^k\big]\\ 
			\leq & (1 - \mu\gamma^k)\|\mathbf{x}^k-\mathbf{1}_n \otimes \bar{\mathbf{x}}^{k}\|^2\\
			 & + \frac{1}{\mu\gamma^k}\mathbb{E}\big[\|\Lambda^k\|^2 \, | \, \mathcal{F}^k\big]\mathbb{E}\big[\|\mathbf{g}^k\|^2 \, | \, \mathcal{F}^k\big].
		\end{aligned}
	\end{equation}
	Since $\Lambda^k$ is a diagonal matrix, $\|\Lambda^k\|^2 = \max\{(\lambda_{il}^k)^2 \, | \, 1\leq i\leq n, 1\leq l\leq d\}$. Given $\mathbb{E}[\lambda_{il}^k]=\lambda_i^k$ and $\mathbb{D} [\lambda_{il}^k]=(\sigma_i^k)^2$, we have
	\begin{equation}\label{Emathbfx2}
		\begin{aligned}
			 & \mathbb{E}\big[\|\Lambda^k\|^2 \, | \, \mathcal{F}^k\big]
			\leq\max_{1\leq i\leq n}(\lambda_i^k)^2+(\sigma_i^k)^2=\max_{1\leq i\leq n}\theta_i^k=\theta_{\max}^k.
		\end{aligned}
	\end{equation}
	Combining \eqref{g}, \eqref{f}, \eqref{Emathbfx1}, and \eqref{Emathbfx2}, we obtain
	\begin{equation}\label{Emathbfx3}
		\begin{aligned}
			 & \mathbb{E}\big[\big(\|\bar{W}^k\|\cdot\|\mathbf{x}^k-\mathbf{1}_n \otimes \bar{\mathbf{x}}^{k}\| + \|\Lambda^k\|\cdot\| \mathbf{g}^k\|\big)^2 \, | \, \mathcal{F}^k\big]\\ 
			\leq & \Big(1 - \mu\gamma^k + \frac{8L^2\theta_{\max}^k}{\mu\gamma^k} \Big) \|\mathbf{x}^{k}-\mathbf{1}_n \otimes \bar{\mathbf{x}}^{k}\|^2\\
			 & +\frac{8 n L^2\theta_{\max}^k}{\mu\gamma^k }\|\bar{\mathbf{x}}^k-\mathbf{x}^*\|^2+\frac{2n\theta_{\max}^k}{\mu\gamma^k }\sigma_g^2\\
			 & +\frac{4\theta_{\max}^k}{\mu\gamma^k }\sum\nolimits_{i=1}^n\|\nabla f_i(\mathbf{x}^*)\|^2.
		\end{aligned}
	\end{equation}
	Substituting \eqref{Emathbfx3} into \eqref{Emathbfx} yields \eqref{part2} for $k \geq K$. \hfill$\blacksquare$
\end{pf}

\begin{theorem}
	Under Assumptions \ref{ass:graph}, \ref{ass:Ff}, and \ref{ass:step}, if the optimization problem \eqref{Fd} admits at least one optimal solution $\mathbf{x}^*$, then our algorithm achieves almost sure convergence. Specifically, each agent's state $\mathbf{x}_i^k$ converges to an optimal solution $\mathbf{x}^*$ of \eqref{Fd}, provided that the sequences $\{\lambda_i^k\}$, $\{\sigma_i^k\}$, and $\{\gamma^k\}$ satisfy the following conditions for all $i \in \mathcal{V}$:
	\begin{equation}\label{lambdaxi}
		\begin{aligned}
			 & \sum\nolimits_{k=0}^{\infty} \lambda_i^k=\infty, \quad \sum\nolimits_{k=0}^{\infty}(\lambda_i^k)^2<\infty, \\
			 & \sum\nolimits_{k=0}^{\infty} \sum\nolimits_{i, j \in \mathcal{V}}|\lambda_i^k-\lambda_j^k|<\infty, \enspace \sum\nolimits_{k=0}^{\infty}(\sigma_i^k)^2<\infty, \\
			 & \sum\nolimits_{k=0}^{\infty} \gamma^k=\infty, \quad \sum\nolimits_{k=0}^{\infty}(\gamma^k)^2<\infty, \\
			 & \sum\nolimits_{k=0}^{\infty}\frac{(\lambda_i^k)^2}{\gamma^k}<\infty, \quad \sum\nolimits_{k=0}^{\infty}\frac{(\sigma_i^k)^2}{\gamma^k}<\infty.
		\end{aligned}
	\end{equation}
\end{theorem}
\begin{pf}
	Combining the results of Lemmas \ref{lemma:first} and \ref{lemma2}, there exists a positive integer $K$ such that for all $k \ge K$, the following matrix inequality holds:
	\begin{equation}
		\begin{aligned}
			 & {\left[\begin{array}{c}
					\mathbb{E}[\|\bar{\mathbf{x}}^{k+1}-\mathbf{x}^*\|^2 \, | \, \mathcal{F}^k] \\
					\mathbb{E}[\|\mathbf{x}^{k+1}-\mathbf{1}_n \otimes \bar{\mathbf{x}}^{k+1}\|^2 \, | \, \mathcal{F}^k]
				\end{array}\right]} \\
			\leq & \left(\left[\begin{array}{cc}
				1 & \frac{\gamma^k}{n} \\
				0 & 1 - \mu\gamma^k
			\end{array}\right]+A^k \right)\left[\begin{array}{c}
				\|\bar{\mathbf{x}}^k-\mathbf{x}^*\|^2 \\
				\|\mathbf{x}^k-\mathbf{1}_n \otimes \bar{\mathbf{x}}^k\|^2
			\end{array}\right] \\
			 & +\mathbf{b}^k-2\bar{\lambda}^k\left[\begin{array}{c}
				F(\bar{\mathbf{x}}^k)-F(\mathbf{x}^*)\\ 
				0
			\end{array}\right], 
		\end{aligned}
	\end{equation}
	where $A^k=\left[\begin{array}{ll}a_{11}^k & a_{12}^k \\ a_{21}^k & a_{22}^k\end{array}\right]$ and $\mathbf{b}^k=\left[\begin{array}{l}b_1^k \\ b_2^k\end{array}\right]$ are defined by
	\begin{equation*}
		\begin{aligned}
			a_{11}^k= & \frac{nL^2(\bar{\lambda}^k)^2}{\gamma^k}+\frac{\sqrt{d}(2 nL^2+1)}{n}\|\boldsymbol{\lambda}^k-\bar{\lambda}^k \mathbf{1}_n\|\\
			 & +\frac{8dL^2\sum_{i=1}^n \theta_i^k}{n}, \\
			a_{12}^k = & \frac{8dL^2\sum_{i=1}^n \theta_i^k}{n^2}, \ \ a_{21}^k=\frac{8 n L^2\theta_{\max}^k}{\mu\gamma^k}, \ \
			a_{22}^k =\frac{8 L^2\theta_{\max}^k}{\mu\gamma^k}, 
		\end{aligned}
	\end{equation*}
	\begin{equation*}
		\begin{aligned}
			b_1^k = & \Big(\frac{4d\sum_{i=1}^n \theta_i^k}{n^2}+\frac{2\sqrt{d}}{n}\|\boldsymbol{\lambda}^k-\bar{\lambda}^k \mathbf{1}_n\|\Big)\sum\nolimits_{i=1}^n\|\nabla f_i(\mathbf{x}^*)\|^2\\
			 & +\frac{1}{2}(n-1)^2d\delta^2(\gamma^k)^2+\frac{2d\sum_{i=1}^n \theta_i^k}{n^2}\sigma_g^2, \\
			b_2^k= & \frac{4\theta_{\max}^k}{\mu\gamma^k} \sum\nolimits_{i=1}^n\|\nabla f_i(\mathbf{x}^*)\|^2+\frac{2n\theta_{\max}^k}{\mu\gamma^k}\sigma_g^2\\
			 & +\frac{1}{2}(n-1)^2d\delta^2(\gamma^k)^2.
		\end{aligned}
	\end{equation*}
	
	Define $a_{\max}^k=\max \{a_{11}^k, a_{12}^k, a_{21}^k, a_{22}^k\}$ and $b_{\max}^k = \max \{b_1^k,$ $b_2^k\}$. Then we have $A^k \leq a_{\max}^k \mathbf{1}_n\mathbf{1}_n^T$ and $\mathbf{b}^k \leq b_{\max}^k \mathbf{1}_n$. 
	
	Under the conditions in \eqref{lambdaxi}, it can be verified that $\sum_{k=0}^{\infty}a_{\max}^k<\infty$, $\sum_{k=0}^{\infty} b_{\max}^k<\infty$, and $\sum_{k=0}^{\infty} \bar{\lambda}^k=\infty$. Therefore, following Lemma 6 in \cite{Wang2024Tailor}, we have $\lim _{k \rightarrow \infty}\|\mathbf{x}_i^k-\bar{\mathbf{x}}^k\|=0$ holds almost surely for all agents and $\lim _{k \rightarrow \infty}\|\bar{\mathbf{x}}^k-\mathbf{x}^*\|=0$ holds almost surely. \hfill$\blacksquare$
\end{pf}

\subsection{Privacy Analysis}

\begin{definition}\label{def:eavesdropper}	
	An eavesdropper, denoted as $\mathcal{A}$, is an external attacker who knows the topology of the network and has the ability to intercept all exchanged messages in order to infer the private gradients of participating agents.
\end{definition}

\begin{definition}\label{def:privacy}
	For a distributed network with $n$ agents, the privacy of the gradient information $\mathbf{g}_i^k$ of agent $i$ is preserved if an external eavesdropper cannot infer the exact value of $\mathbf{g}_i^k$ at any iteration $k$.
\end{definition}

\begin{theorem}
	For a connected graph $\mathcal{G}=(\mathcal{V}, \mathcal{E})$, the proposed algorithm preserves the privacy of each agent $i$ against an external eavesdropper.
\end{theorem}

\begin{pf}
	We define the information $\mathcal{I}_{\mathcal{A}}$ accessible to the external eavesdropper $\mathcal{A}$ as:
	\begin{equation}\label{infor}
		\begin{aligned}
			 & \mathcal{I}_{\mathcal{A}} \triangleq \{ \mathcal{I}_{ij}^{\text {exc }}(k) \, | \, i, j \in \mathcal{V}, \, i\neq j, \, k \geq 0 \}, 
		\end{aligned}
	\end{equation}
	where $\mathcal{I}_{ij}^{\text {exc }}(k)$ is given by:
	\begin{equation*}\label{information}
		\begin{aligned}
			  \mathcal{I}_{ij}^{\text {exc }}(k) 
			 =  \big\{ \mathbb{E}\big[Q(w_{j \rightarrow i})\big(Q(\mathbf{x}_j^k)-Q(\mathbf{x}_i^k)\big)\big], \mathbb{E}[Q(-\mathbf{x}_i^k)] \big\}.
		\end{aligned}
	\end{equation*}
	From $\mathcal{I}_{\mathcal{A}}$, we observe that the eavesdropper cannot infer agent $i$'s gradient $\mathbf{g}_i^k$, because all exchanged messages are encrypted using the Paillier cryptosystem, which hides agent $i$'s state $\mathbf{x}_i^k$ and coupling weights $w_{i \rightarrow j}$.

	Even if the eavesdropper breaks the Paillier encryption, reducing $\mathcal{I}_{ij}^{\text {exc }}(k)$ to the plaintext form
	\begin{equation}
		\begin{aligned}
			 \tilde{\mathcal{I}}_{ij}^{\text {exc }}(k) = \{w_{j \rightarrow i}(\mathbf{x}_j^k-\mathbf{x}_i^k), \, \mathbf{x}_i^k \}, 
		\end{aligned}
	\end{equation}
 where we ignored the quantization error introduced by Paillier encryption. However, even with access to all states $\mathbf{x}_i^k$ and all coupling weights $w_{i \rightarrow j}$, the external eavesdropper $\mathcal{A}$ cannot infer the gradient $\mathbf{g}_i^k$ of agent $i$ by \eqref{x} and \eqref{infor} because it lacks knowledge of the attenuation factor $\gamma^k$ and the stepsize $\Lambda _i^k$. Consequently, the proposed algorithm preserves the privacy of participating agents against the external eavesdropper. \hfill$\blacksquare$
\end{pf}

\begin{remark}
	Following an analogous reasoning, the proposed algorithm also protects agents' privacy against internal honest-but-curious adversaries (individually or collusively), external eavesdroppers, or a coalition of both. Even if the Paillier encryption is broken and all exchanged information is revealed, the attacker only obtains the product $\Lambda_i^k \mathbf{g}_i^k$ and cannot recover the true gradient $\mathbf{g}_i^k$. This is because the heterogeneous and time-varying stepsizes allow each agent $i$ to perform strategic scaling $\{ \tilde{\Lambda}_i^k, \tilde{\mathbf{g}}_i^k \} = \{ e^{-\zeta^k} \Lambda_i^k, e^{\zeta^k} \mathbf{g}_i^k \}$ for any $0<\zeta^k<\infty$, such that the observed product remains unchanged while the true gradient is obfuscated.
\end{remark}

\section{Numerical Simulations}
We consider a distributed estimation problem where each sensor $i$ collects $n_i$ noisy measurements $\mathbf{z}_{ij} = M_i \mathbf{x} + \boldsymbol{\nu}_{ij}$ ($j = 1, \ldots, n_i$), with $M_i \in \mathbb{R}^{p \times d}$ and $\boldsymbol{\nu}_{ij}$ denoting the measurement matrix and associated noise, respectively. The unknown parameter $\mathbf{x}$ is estimated via empirical risk minimization, with local cost $f_i(\mathbf{x}) = \frac{1}{n_i} \sum_{j=1}^{n_i} \| \mathbf{z}_{ij} - M_i \mathbf{x} \|^2 + \omega \| \mathbf{x} \|^2$, where $\omega$ denotes the regularization parameter.

\begin{figure}[h]
	\centering
	\includegraphics[width=0.32\columnwidth]{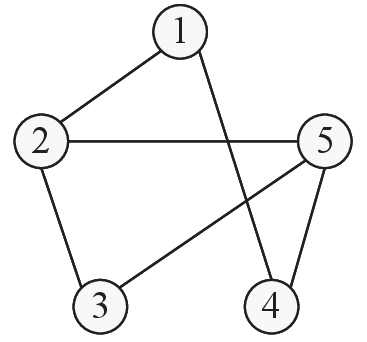}
	\caption{Undirected topology of the network.}
	\label{undirected topology}
\end{figure}

We considered a network of five agents with the undirected topology shown in Fig. \ref{undirected topology}. The parameters were set as $p=3$, $d=2$, and $n_i=50$ for each agent $i$. The measurement noise $\boldsymbol{\nu}_{ij}$ followed a uniform distribution on $[-0.5, 0.5]$. In our proposed algorithm, we set the quantization precision $\delta = 0.1$ and the attenuation factor $\gamma^k = \frac{1}{1+0.1k^{0.81}}$. The stepsize matrix for agent $i$ at iteration $k$ was set as $\Lambda_i^k = \operatorname{diag}\{\lambda_{i1}^k, \ldots, \lambda_{id}^k\}$, where each diagonal entry was generated as $\lambda_{il}^k = \frac{0.005}{k^{0.6}}\big(1 + \frac{\zeta_{il}^k}{k^{1.2}}\big)$ for $l=\{1, \ldots, d\}$, with $\{\zeta_{il}^k\}$ being independent uniform random variables on $[0, 1]$. This configuration satisfies the conditions in \eqref{lambdaxi}.
 
As a performance metric, we tracked the mean squared error $\sum_{i=1}^{n} \| \mathbf{x}_i^k - \mathbf{x}^* \|^2$ and its variance over 1000 trials. The results are shown in Fig. \ref{errorbar}. For comparison, we also implemented a variant of our algorithm in \eqref{x} without the attenuation factor $\gamma^k$ (i.e., $\gamma^k=1$) and the conventional distributed stochastic optimization algorithm in \eqref{eq:SGD} (no encryption or quantization). The red solid, blue dashed, and purple dash-dotted curves in Fig. \ref{errorbar} show the average optimization error for the three methods, respectively. Evidently, the convergence rate attained by our algorithm is on par with that of the conventional algorithm, although with a slight, relatively limited, degradation in final accuracy. In contrast, the version without the attenuation factor exhibits poorer accuracy, demonstrating the importance of $\gamma^k$ in suppressing quantization error.

\begin{figure}[h]
	\centering
	\includegraphics[width=1\columnwidth]{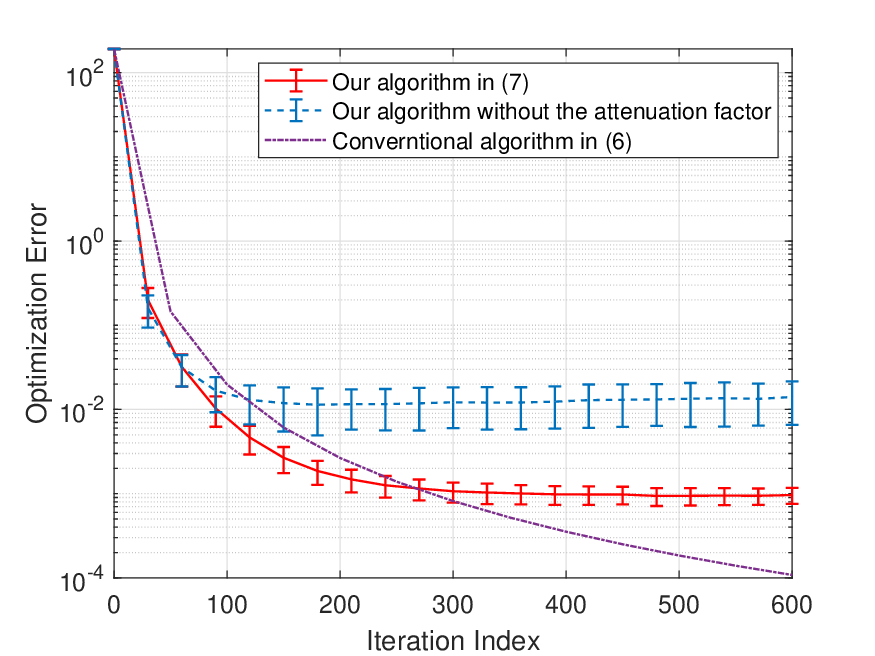}
	\caption{Comparison of our algorithm with conventional distributed stochastic optimization algorithm.}
	\label{errorbar}
\end{figure}

\section{Conclusions}
We propose a distributed stochastic optimization algorithm designed to preserve privacy, utilizing Paillier homomorphic encryption in conjunction with heterogeneous, time-varying step sizes. The encryption ensures confidentiality against eavesdroppers, while the heterogeneous stepsizes mask true gradients even if encryption is compromised. To suppress quantization error introduced by Paillier encryption, an attenuation factor is incorporated in the algorithm, ensuring almost sure convergence without sacrificing privacy. The algorithm requires no trusted neighbors and preserves privacy even when all neighbors collude with external eavesdroppers. Theoretical analysis and numerical simulations validate both the convergence and privacy-preserving properties of our algorithm.

\bibliography{reference}

\end{document}